\documentclass[a4paper]{article}
\usepackage[margin=2cm]{geometry}
\usepackage[T1]{fontenc}
\usepackage{graphicx}
\usepackage[ruled,vlined,oldcommands]{algorithm2e}
\usepackage{tikz}
\usepackage{xspace}
\usepackage{hyperref}

\usepackage{amssymb,latexsym,amsmath,amsthm}
\usepackage{algorithm2e}

\usepackage{verbatim} % for comments

\renewcommand{\P}{\mathbb P}
\newcommand{\K}{\mathbb K}
\newcommand{\Q}{\mathbb Q}

\newcommand{\T}{\mathbb T}
\newcommand{\G}{\mathbb G}
\renewcommand{\H}{\mathbb H}

\newcommand{\Pos}{\operatorname{Pos}}
\newcommand{\isucc}{\operatorname{isucc}}
\newcommand{\up}[2]{{\uparrow}^{#1}#2}

\newcommand{\lshomo}{\textsf{LSHom}\xspace}

\DontPrintSemicolon

\usepackage{todonotes}

\theoremstyle{plain}
\newtheorem{theorem}{Theorem}[section]
\newtheorem{lemma}[theorem]{Lemma}

\newtheorem{observation}[theorem]{Observation}
\newtheorem{corollary}[theorem]{Corollary}

\begin{document}
\title{Tabular intermediate logics comparison\footnote{
This preprint has not undergone peer review or any post-submission improvements or corrections. The Version of Record of this contribution is published in Proceeding of WoLLIC 2025, LNCS 15942~\cite{wollic}, and is available online at \url{https://doi.org/10.1007/978-3-031-99536-1_20}}}

\date{}
\author{Paweł Rzążewski\footnote{Warsaw University of Technology \& University of Warsaw, Poland. The first author was supported by the European Research Council (ERC) under the European Union’s Horizon 2020 research and innovation programme, grant agreement no 948057~(BOBR).} \and Michał Stronkowski\footnote{Warsaw University of Technology, Poland}}

\maketitle              % typeset the header of the contribution
\begin{abstract}
Tabular intermediate logics are intermediate logics characterized by finite posets treated as Kripke frames. For a poset $\mathbb P$, let $L(\mathbb P)$ denote the corresponding tabular intermediate logic. We investigate the complexity of the following decision problem {\sf LogContain}: given two finite posets $\mathbb P$ and $\mathbb Q$, decide whether $L(\mathbb P) \subseteq L(\mathbb Q)$.

By Jankov's and de Jongh's theorem, the problem {\sf LogContain} is related to the problem {\sf SPMorph}: given two finite posets $\mathbb P$ and $\mathbb Q$, decide whether there exists a surjective $p$-morphism from $\mathbb P$ onto $\mathbb Q$.  Both problems belong to the complexity class NP.

We present two contributions. First, we describe a construction which, starting with a graph $\mathbb G$, gives a poset ${\sf Pos}(\mathbb G)$ such that there is a surjective locally surjective homomorphism (the graph-theoretic analog of a $p$-morphism) from $\mathbb G$ onto $\mathbb H$ if and only if there is a surjective $p$-morphism from ${\sf Pos}(\mathbb G)$ onto ${\sf Pos}(\mathbb H)$. This allows us to translate some hardness results from graph theory and obtain that several restricted versions of the problems {\sf LogContain} and {\sf SPMorph} are NP-complete. Among other results, we present a 18-element poset $\Q$ such that the problem to decide, for a given poset $\P$, whether $L(\P)\subseteq L(\Q)$ is NP-complete.

Second, we  describe a polynomial-time algorithm that decides {\sf LogContain} and {\sf SPMorph} for posets $\mathbb T$ and $\mathbb Q$, when $\mathbb T$ is a tree.
\end{abstract}
\section{Introduction}
A (propositional) \emph{intermediate logic} is a logic, i.e., a set of formulas closed under substitutions and {\it Modus Ponens}, lying  between intuitionistic and classical logics. (See e.g. \cite{ChZ97,Mints00} for the necessary information on logic, \cite{Sipser} on the theory of computation, \cite{Diestel} on graph theory and \cite{Har2005} on posets. The essential concepts for this article are presented in Section~\ref{sec:: definitions}.)
Every poset $\P$ determines an intermediate logic $L(\P)$. It consists of formulas that are true in $\P$ when $\P$ is interpreted as a Kripke frame.  If $\P$ is finite, we say that $L(\P)$ is \emph{tabular}. We are interested in comparing such logics with respect to inclusion. In particular, we investigate the computational complexity of the following decision problem
\begin{center}
\begin{tabular}{ | l |  }
 \hline
   \textbf{Problem:} {\sf LogContain} \\ \hline
   \textbf{Instance:} A pair $\P$, $\Q$ of finite posets\\  
 \textbf{Question:} Does it hold that $L(\P) \subseteq L(\Q)$?\\
 \hline
\end{tabular}
\end{center}

In \cite{MS2022} 
it was indicated that {\sf LogContain} is NP-complete even under the restriction that the input posets have depth at most three.
This is the starting point of our investigations: we are interested in understanding what makes instances of {\sf LogContain} hard.

The problem {\sf LogContain} is intrinsically related to $p$-morphisms, as we proceed to explain. 
Let $\P$ and $\Q$ be posets.  A \emph{$p$-morphism} (known also as a bounded morphism or a reduction) between $\P$ and $\Q$ is a mapping $h\colon X^\P\to X^\Q$ satisfying 
\begin{itemize}
\item[(HP)]  for all $x,y\in X^\P$, if $x\leq^\P y$ then $h(x)\leq^\Q h(y)$ (\emph{the homomorphism property});
\item[(BP)] for all $x\in X^\P$ and $y\in X^\Q$, if $h(x)\leq^\Q y$ then there exists $z\in X^\P$ such that $x\leq^\P z$ and $h(z)=y$ (\emph{the backward property}).
\end{itemize}

Let us recall crucial Jankov's and de Jongh's theorem \cite{Jankov63,deJongh68,DeJY11}. Here ${\uparrow}^\P x$, for an element $x$ of a poset $\P$, denotes the upset in $\P$ generated by $x$.  

\begin{theorem}
\label{thm:: Jankov}
For every finite rooted poset $\P$ there exists a formula $\chi_\P$ such that, for each poset $\Q$, we have $\chi_\P\not\in L(\Q)$ if and only if  $\P$ is a $p$-morphic image of ${\uparrow}^\Q x$ for some $x\in X^\Q$.
\end{theorem}

Theorem~\ref{thm:: Jankov} allows us, while studying {\sf LogContain},  to focus only on $p$-morphisms, as the following direct corollary shows.
 
\begin{corollary}
\label{cor:: from Jankov}
Let  $\P$ and $\Q$ be finite posets. Then $L(\P)\subseteq L(\Q)$ if and only if every poset ${\uparrow}^\Q y$, where $y$ is minimal in $\Q$, is a $p$-morphic image of ${\uparrow}^\P x$ for some $x\in X^\P$. In particular, if both $\P$ and $\Q$ are rooted and have the same depth, then $L(\P)\subseteq L(\Q)$ if and only if there is a surjective $p$-morphism from $\P$ onto~$\Q$.  
\end{corollary}

The above considerations lead us to the following decision problem:

\begin{center}
\begin{tabular}{ | l |  }
 \hline
   \textbf{Problem:} {\sf SPMorph} \\ \hline
   \textbf{Instance:} A pair $\P$, $\Q$ of finite posets\\
 \textbf{Question:} Does there exist a surjective $p$-morphism  $h\colon \P\to\Q$?\\
 \hline
\end{tabular}
\end{center}

It is evident that {\sf SPMorph} belongs to the complexity class NP. By Corollary \ref{cor:: from Jankov}, the same holds for {\sf LogContain}. 

Graph-theoretic analogs of $p$-morphisms were studied in social sciences as early as the 1980s and 1990s, see e.g. \cite{EB91,PR2001,WR83}. They appeared under various names, like \emph{locally surjective homomorphisms}, \emph{role assignments} or \emph{role colorings}. Computational aspects of locally surjective (and related locally bijective) graph homomorphisms were studied in e.g.~\cite{BJK2011,Bulteau2024,CHAPLICK2015,FP2005,Kra94,PR2001}.
(See e.g. \cite{BS14} for historical information on $p$-morphisms.)

Our main contribution in this paper is the constructions that, for a graph $\G$, produces a poset $\Pos(\G)$ and a rooted poset $\Pos_\bot(\G)$, each of size linear in the size of $\G$, such that the following conditions are equivalent (see Theorem~\ref{thm:: main corresponcence} and Corollary~\ref{cor:: main corresponcence with bottom}):
\begin{itemize}
\item  There exists a surjective locally surjective homomorphism from $\G$ onto $\H$;
\item  There exists a surjective $p$-morphism from $\Pos(\G)$ onto $\Pos(\H)$;
\item  There exists a surjective $p$-morphism from $\Pos_\bot(\G)$ onto $\Pos_\bot(\H)$.
\end{itemize}
This allows us to transfer some hardness results from graphs to posets. We summarize them below. In what follows, by $(\mathbb{P}\mathrm{ath}_2$ we denote the three-vertex path and by $\K_4$ we denote the complete graph on four vertices. The Hasse diagram of $\Pos_\bot(\mathbb{P}\mathrm{ath}_2)$ is depicted in Figure \ref{fig} (some edges are plotted as dashed to improve visibility). 

\begin{theorem}
The problem {\sf LogContain} is NP-complete even for inputs $\P,\Q$, where
\begin{enumerate}
\item  $\P$ is of depth 5 and $\Q = \Pos_\bot(\mathbb{P}\mathrm{ath}_2)$ {\rm(}in particular, $\Q$ has 18 elements{\rm)};
\item  $\P$ is of dimension at most 7, and its every element except for the root has at most 4 immediate successors and at most 12 successors in all, and \mbox{$\Q = \Pos_\bot(\K_4)$} {\rm(}in particular, $\Q$ has 29 elements{\rm)};
\item $\P$ has pathwidth at most 20 and $\Q$ has pathwidth at most 17.
\end{enumerate}
\end{theorem}

\begin{theorem}
The problem {\sf SPMorph} is NP-complete even for inputs $\P,\Q$, where
\begin{enumerate}
\item  $\P$ is of depth 4 and $\Q = \Pos(\mathbb{P}\mathrm{ath}_2)$ {\rm(}in particular $\Q$, has 17 elements{\rm)};
\item  $\P$ is of dimension at most 7, and its every element has at most 4 immediate successors and at most 12 successors in all; and $\Q = \Pos(\K_4)$ {\rm(}in particular, $\Q$ has 28 elements{\rm)};
\item $\P$ has pathwidth at most 19 and $\Q$ has pathwidth at most 16.
\end{enumerate}
\end{theorem}

\begin{figure}[h]
\centering
\begin{tikzpicture}[scale=0.85]
    \tikzstyle{point} = [shape=circle, thick, draw=black, fill=black , scale=0.42]

\node (bot) at (0,-4) [point,label=below:{$\bot$}] {};

\node (v1) at (-4,-2) [point, label=below left:{$u$}] {};
\node (v2) at (0,-2) [point, label=below left:{$v$}] {};
\node (v3) at (4,-2) [point,label=below right:{$w$}] {};

\node (v1c) at (-5,0) [point,label=left:{$u_a$}] {};
\node (v1d) at (-3,0) [point,label=above left:{$u_b$}] {};
\node (v2c) at (-1,0) [point,label=below left:{$v_b$}] {};
\node (v2d) at (1,0) [point,label=below right:{$v_a$}] {};
\node (v3c) at (3,0) [point,label=above right:{$w_a$}] {};
\node (v3d) at (5,0) [point,label=right:{$w_b$}] {};

\node (e1) at (-3,2) [point,label=left:{$e_1$}] {};
\node (e1b) at (-1,2) [point,label=left:{$e_2$}] {};

\node (e2) at (1,2) [point,label=right:{$f_1$}] {};
\node (e2b) at (3,2) [point,label=right:{$f_2$}] {};

\node (T1) at (-2,4) [point,label=above:{$\top_1$}] {};
\node (T2) at (2,4) [point,label=above:{$\top_2$}] {};

\node (infc) at (-4,4) [point,label=above:{$\infty_a$}] {};
\node (infd) at (4,4) [point,label=above:{$\infty_b$}] {};

\draw (bot) -- (v1) ;
\draw (bot) -- (v2) ;
\draw (bot) -- (v3) ;
\draw (v1c) -- (v1) -- (v1d);
\draw (v1c) -- (e1) -- (v2c);
\draw (v1d) -- (e1b) -- (v2d);
\draw (v2c) -- (v2) -- (v2d);
\draw (v2c) -- (e2) -- (v3c);
\draw (v2d) -- (e2b) -- (v3d);

\draw (v3c) -- (v3) -- (v3d);
\draw (e1) -- (T1) -- (e2) -- (T2) -- (e1);
\draw (e1b) -- (T1) -- (e2b) -- (T2) -- (e1b);
\draw (v1c) -- (infc) -- (v2c) [dashed];
\draw (v3c) -- (infc) [dashed];
\draw (v1d) -- (infd) -- (v2d) [dashed];
\draw (v3d) -- (infd) [dashed];

\end{tikzpicture}
\caption{The rooted poset $\Pos_\bot(\mathbb{P}\mathrm{ath}_2)$ for the three-vertex path $\mathbb{P}\mathrm{ath}_2$.} %end of caption
\label{fig}
\end{figure}

We finish the paper with one positive result. Namely, we present a polynomial-time algorithm that decides {\sf LogContain} and {\sf SPMorph} for posets $\T$ and $\Q$, when $\T$ is a tree.

\section{Preliminaries}
\label{sec:: definitions}
Here we provide basic definitions primarily  in order to fix notation. 
A \emph{poset} $\P$ is a structure consisting of the \emph{carrier} set $X^\P$ and the order relation $\leq^\P$. Its  \emph{covering relation} is given by
\[
\prec^\P\;=\{(x,y): x,y\in X^\P,\text{ } x<^\P y,\text{ there is no }z\in X^\P \text{ such that }x<^\P z<^\P y \}.
\]
For an element $x$ of a poset $\P$ the sets ${\uparrow}^\P x=\{y\in X^\P : x\leq ^\P y\}$ and ${\downarrow}^\P x=\{y\in X^\P : y\leq ^\P x\}$ are called (\emph{principal}) \emph{upset} and \emph{downset} respectively. Clearly, both these sets are carriers of subposets of $\P$ with the order inherited from $\P$. We do not distinguish notationally these posets from their carriers.   A poset $\P$ is \emph{rooted} if there exists $r\in X^\P$ such that $\P={\uparrow}^\P r$. The element $r$ is then called the \emph{root} of $\P$. 

We say that an element $x$ of a poset $\P$ is \emph{of depth $n$} if ${\uparrow}x$ contains a chain of cardinality $n$ but not a larger one. 
The \emph{depth} of a poset is the maximum depth of its elements.
We will use the following basic observation that follows from (BP). 

\begin{observation}
\label{obs}
Let $h\colon \P\to\Q$ be a $p$-morphism. Then,
for $x\in X^\P$, the depth of $h(x)$ is not greater than the depth of $x$. Also the number of elements in $\up{\Q}{h(x)}$ is not greater than the number of elements in  $\up{\P}{x}$.
In particular, $h$ maps maximal elements of $\P$ into maximal elements of $\Q$.
\end{observation}

Let us give a sketch of the proof of Corollary \ref{cor:: from Jankov}.
\begin{proof}[Proof of Corollary \ref{cor:: from Jankov}]
Suppose that $L(\P)\subseteq L(\Q)$ and let $y$ be a minimal element of $\Q$. Let us consider a formula $\chi_{{\uparrow}^\Q y}$. Since the identity mapping on ${\uparrow}^\Q y$ is a $p$-morphism, by Theorem~\ref{thm:: Jankov}, $ \chi_{{\uparrow}^\Q y}\not\in L(\Q)$. Hence, $\chi_{{\uparrow}^\Q y}\not\in L(\P)$ and, by Theorem~\ref{thm:: Jankov}, there exists $x\in X^\P$ such that ${\uparrow}^\Q y$ is a $p$-morphic image of ${\uparrow}^\P x$.  

For the inverse implication, note that for any formula  $\varphi$, we have $\varphi\in L(\P)$ iff $\varphi\in L({\uparrow}^\P x)$ for all $x\in X^\P$, and that $p$-morphisms preserve the satisfaction of formulas, see \cite[Sec. 2.3]{ChZ97}. The last statement follows from Observation~\ref{obs}.
\end{proof}

A \emph{graph} $\G$ is a structure consisting of the set $V^\G$ of its \emph{vertices} and a set $E^\G$ of its \emph{edges}. An edge is a two-element subset of $V^\G$. The edge $\{u,v\}$ is denoted simply by $uv$.  For graphs $\G$ and $H$ a mapping $g\colon V^\G\to V^\H$ is called a \emph{locally surjective homomorphism} if it satisfies 
\begin{itemize}
\item[(HP)]  for all $u,v\in V^\G$, if $uv\in E^\G$ then $g(u)g(v)\in E^\H$ (\emph{the homomorphism property});
\item[(BP)] for all $u\in V^\G$ and $w\in V^\H$, if $g(u)w\in E^\H$ then there exists $v\in V^\G$ such that $uv\in E^\G$ and $g(v)=w$ (\emph{the backward property}).
\end{itemize}

\section{From graphs to posets} \label{sec:construction}

Our aim in this section is to provide a construction of the poset $\Pos(\G)$ for a given graph $\G$, such that there exists a surjective locally surjective homomorphism from $\G$ onto $\H$ if and only if there is a surjective $p$-morphism from $\Pos(\G)$ onto $\Pos(\H)$.    
The construction of $\Pos(\G)$ is inspired by the incidence poset of a graph~\cite{Sch89}.
For the application in logic, we will also consider a modification of $\Pos(\G)$. Let $\Pos_\bot(\G)$ be obtained from $\Pos(\G)$ by adding a root $\bot$ to it.
% Notice that $\P(\G)$ is rooted if it has only one vertex. In such a case, $\bot$ is a fresh root.
The poset $\Pos_\bot(\mathbb{P}\mathrm{ath}_2)$ for the three-element path graph $\mathbb{P}\mathrm{ath}_2$, with vertices $u,v,w$ and edges $e=uv,f=vw$, is depicted in Figure \ref{fig}.

Let $\G=(V^\G,E^\G)$ be a graph. The carrier of $\Pos(\G)$ is given by  
\[
X^{\Pos(\G)}= V^\G\cup V^\G_a\cup V^\G_b \cup E^\G_1\cup E^\G_2\cup \{\top_1,\top_2,\infty_a,\infty_b\}, 
\]
where $V^\G_a$, $V^\G_b$ are copies of the vertex set $V^\G$ and $E^\G_1$, $E^\G_2$ are copies of the edge set  $E^\G$. 
%The subscript $d$ refers to {\it domain} and $c$ refers to {\it codomain}. 
The copies of $v\in V^\G$ in $V_a^\G$ and in $V_b^\G$ are denoted by $v_a$ and $v_b$ respectively. Similarly, $e_1$ and $e_2$ are copies of $e\in E^\G$ in $ E^\G_1$, $ E^\G_2$ respectively. The purpose of the auxiliary elements $\top_1,\top_2,\infty_a,\infty_b$ is to {\it improve rigidity}, in a sense of eliminating unwanted $p$-morphisms between posets $\Pos(\G)$ and $\Pos(\H)$. Let

\begin{align*}
\triangleleft  &= \{(v,v_a),(v,v_b) : v\in V^\G\}\\
%&\cup\{(u_a,e_a),(v_a,e_a),(u_b,e_b),(v_b,e_b) : e=uv\in E^\G\}\\
& \cup \{(v_a,\infty_a),(v_b,\infty_b) : v\in V^\G\}\\
&\cup \{(e_1,\top_1),(e_1,\top_2),(e_2,\top_1),(e_2,\top_2) : e\in E^\G\}\\
&\cup \{(v_a,\top_1), (v_b,\top_1),(v_a,\top_2), (v_b,\top_2)  : v\in V^\G \text{ is isolated in }\G \} 
\end{align*}
The covering relation of $\Pos(\G)$ is the expansion of $\triangleleft$ obtained by adding, for every $e=uv\in E^\G$, the covers $(u_a,e_i), (v_b,e_i)$ and  $(u_b,e_j), (v_a,e_j)$ in such a way that $\{i,j\}=\{1,2\}$.

\begin{theorem}
\label{thm:: main corresponcence}
Let $\G$, $\H$ be graphs such that $\H$ has at least one vertex, and let $\P=\Pos(\G)$, $\Q=\Pos(\H)$. Then there exists a surjective locally surjective homomorphism from $\G$ onto $\H$ if and only if there is a surjective $p$-morphism from $\P$ onto $\Q$.
\end{theorem}  

The proof of Theorem~\ref{thm:: main corresponcence} is organized as follows. We start with consideration of a surjective $p$-morphism  $h\colon \P\to \Q$. We show that $h$ maps $V^\G$ onto $V^\H$ and the restriction of $h$ to $V^\G$ is a surjective locally surjective homomorphism from $\G$ onto $\H$.  Then we show that every surjective locally surjective homomorphism $g\colon \G\to\H$ extends to a surjective $p$-morphism from $\P$ onto $\Q$.
We split the reasoning into a series of lemmas.

\begin{lemma}
\label{lem:: top aux}
We have $h(\{\top_1,\top_2\})\subseteq\{\top_1,\top_2\}$.
\end{lemma}

\begin{proof}
We consider only $\top_1$. By Observation~\ref{obs},  $h(\top_1)\in\{\top_1,\top_2,\infty_a,\infty_b\}$. Suppose, for the sake of contradiction and without loss of generality, that \mbox{$h(\top_1)=\infty_a$}. Then, by (HP), $h({\downarrow}^\P\top_1)\subseteq{\downarrow}^\Q\infty_a$. %=V^\H\cup V^\H_a\cup\{\infty_a\}$. 
Since we assumed that $\H$ has at least one vertex, say $u$, the four-element set $\{\top_1,\top_2,\infty_b,u_b\}$ is contained in $X^\Q\setminus {\downarrow}^\Q\infty_a$. And since $X^\P\setminus{\downarrow}^\P\top_1=\{\top_2,\infty_a,\infty_b\}$ has three elements, we reach a contradiction with the surjectivity of $h$. 

% uwaga: Jeśli G jest pusty to P ma 4 elementy - tylko maksymalne. wydawaloby sie wtedy, ze top moze isc tez na infty. Ale taka sytuacja nie zaistnieje poniewaz niepustosc H i surjektywnosc h pociaga, ze i G nie moze byc pusty.

\end{proof}

\begin{lemma}
\label{lem:: infty}
We have
$h(\{\infty_a,\infty_b\})=\{\infty_a,\infty_b\}$.
\end{lemma}

\begin{proof} 
By Lemma~\ref{lem:: top aux} and (HP), 
\[
h(X^\P\setminus\{\infty_a,\infty_b\})=h({\downarrow}^\P\top_1\cup {\downarrow}^\P\top_2)\subseteq {\downarrow}^\Q\top_1\cup {\downarrow}^\Q\top_2=X^\H\setminus\{,\infty_a,\infty_b\}.
\]
Thus the statement follows by the surjectivity of $h$. 
\end{proof}

\begin{lemma}
\label{lem:: top}
We have $h(\{\top_1,\top_2\})=\{\top_1,\top_2\}$.
\end{lemma}

\begin{proof}
By the  surjectivity of $h$, there are $x_1,x_2\in X^\P$ such that $h(x_1)=\top_1$ and $h(x_2)=\top_2$. By Lemma~\ref{lem:: infty}, $x_1,x_2\not\in\{\infty_a,\infty_b\}$.  Thus  $x_1\leq^\P \top_i$ and $x_2\leq^\P \top_j$ for some $i,j\in\{1,2\}$. But then, by (HP), we infer that $h(\top_i)=\top_1$ and $h(\top_j)=\top_2$. This forces that $i\neq j$. Thus the lemma follows. 
% 
%Suppose that the statement is false. Then, by Lemma~\ref{lem:: top aux},  supposition $h(\top_1)=h(\top_2)=\top_i$ implies that $\top_{3-i}$ is not in the image of $h$.  
\end{proof}

\begin{lemma}
\label{lem:: E}
We have $h(E^\G_1\cup E^\G_2)\subseteq E^\H_1\cup E^\H_2$.
\end{lemma}

\begin{proof}
Let $e_1\in E^\G_1$. By Observation~\ref{obs}, $h(e_1)\not\in V^\H \cup V_b^\H\cup V_a^\H$. Indeed, for every $y\in V^\H \cup V_b^\H\cup V_a^\H$,
% since $G$ has no isolated vertices ({\color{red}Check if it is necessary, maybe not if we consider, instead $\infty$, two extra elements $\infty_a$ and $\infty_b$.}),  
the upset ${\uparrow}^\Q y$, has at least four elements $y$, $\top_1,\top_2$ and $\infty_a$ or $\infty_b$, while ${\uparrow}^\P e_1$ has only three elements $e_1,\top_1,\top_2$.
Suppose now
%, for the sake of contradiction, 
that  $h(e_1)$ is maximal in $\Q$.  Then, by (HP), we would have $h(\top_1)=h(\top_2)=h(e_1)$. This contradicts Lemma~\ref{lem:: top}. The argument for $e_2\in E^\G_2$ is the same. 
\end{proof}

\begin{lemma}
\label{lem:: Vc,Vb}
If $h(\infty_a)=\infty_a$ then $h(V_a^\G)\subseteq V_a^\H$ and $h(V_b^\G)\subseteq V_b^\H$. % and $h(E^\G_a)\subseteq E^\H_a$, $h(E^\G_b)\subseteq E^\H_b$.
If $h(\infty_a)=\infty_b$ then $h(V_a^\G)\subseteq V_b^\H$ and $h(V_b^\G)\subseteq V_a^\H$.% and $h(E^\G_a)\subseteq E^\H_b$, $h(E^\G_b)\subseteq E^\H_a$. 
\end{lemma}

\begin{proof}
Suppose that $h(\infty_a)=\infty_a$.
Observe that $V_a^\G\subseteq {\downarrow}^\P \top_1\cap {\downarrow}^\P \infty_a$. Thus, by (HP) and Lemma~\ref{lem:: top aux} and the assumption, for some $i\in\{1,2\}$,
\[
h(V_a^\G)\subseteq {\downarrow}^\P h(\top_1)\cap {\downarrow}^\P h(\infty_a)\subseteq {\downarrow}^\Q \top_i\cap {\downarrow}^\Q \infty_a=V^\H\cup V^\H_a 
\]

Suppose, for the sake of contradiction, that $h(v_a)\in V^\H$ for some $v\in V^\G$. Then, by (BP), there would exist $x\in{\uparrow}^\P v_a\setminus\{v_a\}\subseteq E_1^\G\cup E_2^\G\cup\{\infty_a,\top_1,\top_2\}$ such that $h(x)=\infty_b$. But, by Lemmas~\ref{lem:: top aux} and~\ref{lem:: E}, and the assumption that $h(\infty_a)=\infty_a$, it is not the case.  Thus $h(V_a^\G)\subseteq V_a^\H$. We infer analogically that  $h(V_b^\G)\subseteq  V_b^\H$. 
%Also, for $e=uv\in E^\G$, we have $u_a\leq^\P e_a$. By (HP) and Lemma \ref{lem:: E}, $V_a^\H \ni h(u_a)\leq^\Q h(e_a)\in E_a^\H\cup E_b^\H$. Thus, by the construction of $\Q$, we have  $h(e_a)\in E_a^\H$. Analogically, $h(e_b)\in E_b^\H$.

The argument in the case when $h(\infty_a)=\infty_b$ is the same. 
\end{proof}

\begin{lemma}
\label{lem:: V}
We have $h(V^\G)\subseteq V^\H$. Moreover, for $v \in V^\G$: if $f(\infty_a)=\infty_a$ then $h(v_a)=h(v)_a$ and $h(v_b)=h(v)_b$, and if $f(\infty_a)=\infty_b$ then $h(v_a)=h(v)_b$ and $h(v_b)=h(v)_a$.   
\end{lemma}

\begin{proof}
By Lemma~\ref{lem:: infty}, $h(\infty_a)\in\{\infty_a,\infty_b\}$. We consider only the situation when $h(\infty_a)=\infty_a$.
Let $v\in V^\G$. Then $v\leq^\P v_a,v_b$. Thus, by Lemma~\ref{lem:: Vc,Vb} and (HP), $h(v)\leq^\Q h(v_a)\in V^\H_a$ and 
$h(v)\leq^\Q h(v_b)\in V^\H_b$. By the definition of $\Q$, there exists at most one element  $w$, which is in $V^\H$, such that $w\leq^\Q h(v_a),h(v_b)$. It follows that $w$ exists and, in fact, $h(v)=w$. Moreover, by the definition of $\Q$, we have $h(v_a)=h(v)_a$ and $h(v_b)=h(v)_b$. 
\end{proof}

Lemma~\ref{lem:: V} yields that the mapping $g\colon V^\G\to V^\H$, given by $g(v)=h(v)$ is well-defined.

\begin{lemma}
\label{lem:: g}
The mapping $g$ is a surjective locally surjective homomorphism from $\G$ onto $\H$.
\end{lemma} 

\begin{proof}
By Lemma~\ref{lem:: infty}, $h(\infty_a)\in\{\infty_a,\infty_b\}$. We consider only the case when $h(\infty_a)=\infty_a$.

By Lemmas~\ref{lem:: top aux},~\ref{lem:: infty},~\ref{lem:: E},~\ref{lem:: Vc,Vb}, and~\ref{lem:: V}, the surjectivity of $h$ yields that $h(V^\G)=V^\H$. Thus $g$ is surjective.

Let us observe that, for an edge $e=uv\in E^\G$,  we have $\{h(e_1),h(e_2\}=\{(g(u)g(v))_1,(g(u)g(v))_2\}$. In particular, $g(u)g(v)$ is an edge in $\H$. For this aim, let us suppose that  $u_a,v_b<^\P e_1$ and $u_b,v_a<^\P e_2$.
By Lemma~\ref{lem:: E} and (HP), we have $h(u_a),h(v_b) \leq^\Q h(e_1)\in E^\H_1\cup E^\H_2$ and $h(u_b),h(v_a) \leq^\Q h(e_2)\in E^\H_1\cup E^\H_2$. By Lemma \ref{lem:: Vc,Vb} and the assumption that $h(\infty_a)=\infty_a$, we have
$h(u)_a,h(v)_b \leq^\Q h(e_1)$ and $h(u)_b,h(v)_a \leq^\Q h(e_2)$.
Since $V_a^\H\cap V_b^\H=\emptyset$, we have  $h(u)_a\neq h(v)_b$ and $h(u)_b\neq h(v)_a$.  
Note that elements $w_a$ and $t_b$, from the sets $V_a^\H$ and  $V_b^\H$ respectively, have a common bound $f_i \in E^\H_1\cup E^\H_2$ if and only if $f=wt$. Thus the claim follows.

It follows directly from the above observation that $g$ is a homomorphism. In order to see that $g$ also satisfies (BP) for graphs, suppose, for some $u\in V^\G$ and $w\in V^\H$, that $f=h(u)w$ is an edge in $E^\H$. Then, by Lemma~\ref{lem:: V}, $h(u_a)=h(u)_a\leq^\Q f_j$ for some $j\in\{1,2\}$. By (BP), there exists $x\in {\uparrow}^\P u_a$ such that $h(x)=f_i$. Notice that, by Lemmas~\ref{lem:: top aux},~\ref{lem:: infty},~\ref{lem:: Vc,Vb} and~\ref{lem:: V}, $x=e_i$ for an edge  $e$ in $\G$ and an index $i\in\{1,2\}$.  We infer that $e=uv$, for some $v\in V^\G$, and, by the observation, $f_j=h(e_i)=(h(u)h(v))_j$. Thus $h(v)=w$. This shows that $g$ is locally surjective.  
\end{proof}

%At this point, it is worth mentioning why we consider three copies of every vertex The reason is that, for a surjective $p$-morphism $h'\colon\P\setminus (V^\G_a\cup V^\G_b)\to\P\setminus (V^\H_a\cup V^\H_b)$, we may have $h'(u)=h'(v)$ for some edge $uv$ in $\G$. It would invalidate the homomorphicity of $g$. Moreover, having only one additional copy of a vertex would result in a problem with the preservation of depth.   

With Lemma~\ref{lem:: g}, we finish the first part of the proof of Theorem~\ref{thm:: main corresponcence}. Now, we consider a surjective locally surjective homomorphism $g\colon \G\to \H$ and show that it may be extended to a surjective $p$-morphism $h\colon \P\to\Q$.
We define $h$ as follows
\[
h(x)=
\begin{cases}
x & \text{ if } x\in\{\top_1,\top_2,\infty_a,\infty_b\},\\
g(u)  &\text{ if } x=u\in V^\G,\\
g(u)_a & \text{ if } x=u_a\text{ for }u\in V^\G,\\
g(u)_b  &\text{ if } x=u_b\text{ for }u\in V^\G,\\
(g(u)g(v))_ j& \text{ if } x=(uv)_i \in E^\G \text{, } u_a,v_b<^\P (uv)_i\\
& \phantom{sd,}\text{and } g(u)_a,g(v)_b<^\Q (g(u)g(v))_j
%(g(u)g(v))_b & \text{ if } x=uv\in E^\G_b.
\end{cases}
\]

\begin{lemma}
\label{lem:: h is surjective p-morphisms}
The mapping $h$ is a surjective $p$-morphism.
\end{lemma}

\begin{proof}
In what follows, we consider points that might seem not straightforward. 

\noindent{\it Checking that $h$ is surjective}: We verify that $f_j$, for $f\in E^\H$ and $j\in \{1,2\}$, is in the range of $h$. 
Let $f=wt$. Then either $w_a,t_b<^\Q f_j$ or $w_b,t_a<^\Q f_j$. Without loss of generality, we assume that $w_a,t_b<^\Q f_j$.   
By the surjectivity of $g$, there exists $u\in V^\G$ such that $g(u)=w$. By (BP) for graphs, there exists $v\in V^\G$ such that $g(w)=t$ and $e=uv\in E^\G$.  We have $u_a,v_b<^\P e_i$  for some $i\in\{1,2\}$. Then $h(e_i)=f_j$ by the definition of $h$.

\noindent{\it Checking that $h$ is a homomorphism}:
We verify that $h(x)\leq^\Q h(y)$ when $x=u_a$ and $y=e_i$, where $e=uv\in E^\G$ and $u_a,v_b<^\P e_i$. Then $h(u_a)=g(u)_a$ and $h(e_i)=(g(u)g(v))_j$, where $g(u)_a,g(v)_b<^\Q (g(u)g(v))_j$. Hence, $h(u_a)<^\Q h(e_i)$.

\noindent{\it Checking that $h$ satisfies} (BP): Suppose that $h(x)\leq^\Q y$. We consider the case when $x=u_a\in V_a^\G$ and $y=(g(u)w)_j\in E^\H_j$.  Then also $h(u_a)=g(u)_a$ and $w_b<^\Q (g(u)w)_j$. We have $g(u)w\in E^\H$ and,
by (BP) for graphs, there exists $v\in V^\G$ such that $g(v)=w$ and $uv\in E^\G$. Hence $u_a,v_b<^\P (uv)_i$, for some $i\in\{1,2\}$, and $h((uv)_i)=(g(u)w)_j$. 
\end{proof}

This finishes the proof of Theorem \ref{thm:: main corresponcence}.
We have the following consequence for rooted posets.

\begin{corollary}
\label{cor:: main corresponcence with bottom}
Let $\G$, $\H$ be graphs such that $\H$ has at least one vertex. Then there exists a surjective locally surjective homomorphism from $\G$ onto $\H$ if and only if there is a surjective $p$-morphism from $\Pos_\bot(\G)$ onto $\Pos_\bot(\H)$.
\end{corollary}
The proof is straightforward and is omitted.

\section{Hardness results}

The constructions in Section~\ref{sec:construction} allow us to translate certain hardness results from the world of graphs to the world of posets.
Let \lshomo denotes the problem where we are given two graphs $\G$ and $\H$, and we need to decide if there exists a locally surjective homomorphism  $g$ from $\G$ to $\H$. In the recalled results, $\H$ will be connected. Thus, if $g$ exists, it is surjective. 

Note that the posets constructed in Section~\ref{sec:construction} have depth (i.e., the number of elements in a longest chain) bounded by 4 in case of $\Pos(\G)$ and by 5 in case of $\Pos_\bot(\G)$.

Fiala and Paulusma~\cite[Theorem 1]{FP2005} proved that \lshomo remains NP-hard even if $\H$ is any fixed connected graph on at least 3 vertices, such as a three-vertex path $\P\rm{ath}_2$.
Note that $\Pos(\P\rm{ath}_2)$ has 17 elements and $\Pos_\bot(\P\rm{ath}_2)$ has 18 elements. Thus, we immediately obtain the following result.

\begin{corollary}
The problem {\sf LogContain} {\rm(}resp., {\sf SPMorph}{\rm)} is NP-complete for inputs $\P$, $\Q$ of depth 5 {\rm(}resp., 4{\rm)} and $\Q = \Pos(\P\rm{ath}_2)$ {\rm(}resp.  $\Q = \Pos_\bot(\P\rm{ath}_2)${\rm)}. In particular, $\Q$ is of constant size.
\end{corollary}

The structure of $\P$ can be restricted even further. Kratochv\'il  showed in \cite[Corollary 4.2]{Kra94} that \lshomo is NP-complete for inputs $\G$, $\K_4$, where $\K_4$ is the complete graph on four vertices, and $\G$ is three-regular (i.e., every vertex is of degree 3). This result was improved by
B\'ilka et al. in \cite[Theorem 2]{BJK2011} by showing that we may restrict to a planar three-regular graph $\G$.

It is straightforward to verify that the construction in Section~\ref{sec:construction} preserves the boundedness of degrees (with the exception for the element $\bot$ of $\Pos_\bot(\G)$.

\begin{lemma}
Let $\G$ be a graph. 
If every vertex of $\G$ has degree at most $k$, where $k\geq 2$,
% nie k\geq 1 bo z powodu izolowanego wierzcholka 
then every element of $\Pos(\G)$ has at most $k+1$ immediate successors, and at most $2k+6$ successors in total.
\end{lemma}

Next, let us exploit the planarity of $\G$. For information on the dimension of a poset, see e.g. \cite[Chap. 7]{Har2005}.

\begin{lemma}
Let $\G$ be a finite planar graph, then the dimension of $\Pos(\G)$ and of $\Pos_\bot(\G)$ is at most 7.
\end{lemma}

\begin{proof}
We observe that the subposet of $\Pos(\G)$ induced by $V^\G \cup E_1^\G$ or, symmetrically, by $V^\G\cup E_2^\G$ is isomorphic to the \emph{incidence poset} of $\G$. As shown by Schnyder~\cite{Sch89}, the incidence poset of a planar graph has dimension at most 3.
Let $\leq_1,\leq_2,\leq_3$ be linear extensions of $\leq\;=\;\leq^{\Pos(\G)}\cap (V^\G \cup E_1^\G)^2$ whose intersection is $\leq$. 

We expand $\leq_1$ and $\leq_2$ as follows.
Each element $v$ is replaced by three elements (in order): $v \prec v_b \prec v_a$, and each element $e_1$ by two elements $e_2\prec e_1$.
Next, we add the elements $\top_1 \prec \top_2 \prec \infty_a \prec \infty _b$ as the four largest elements in the order.
Next, we expand $\leq_3$ similarly.
Each element $v$ is replaced by three elements (in order): $v \prec v_a \prec v_b$ and each element $e_1$ by two elements $e_1\prec e_2$.
Next, we add the elements $\infty_b \prec \infty_a \prec \top_2 \prec \top_1$ as the four largest elements in that order.
Let $\leq'_1,\leq'_2$ and $\leq'_3$ be the obtained linear extensions of $\leq^{\Pos(\G)}$.

Note that their intersection is {\it almost} $\leq^{\Pos(\G)}$, with the exception that the element $\infty_a$ is larger than all elements in $E_1^\G\cup E_2^\G \cup V^\G_b$, and $\infty_b,$ is larger than all elements in $E_1^\G\cup E_2^\G \cup V^\G_a$, and $u_b,v_a\leq e_i$ when $u_a,v_b\leq e_i$ (while they should be incomparable). 
We remedy this by adding four new linear extensions $\leq_4$, $\leq_5$, $\leq_6$ and $\leq_7$. In $\leq_4$ we have 
$(\text{elements in }V^\G)<_4(\text{elements in }V^\G_a)<_4\infty_a<_4(\text{elements in }V^\G_b)<_4\infty_b<_4 (\text{elements in }E_1^\G\cup E_2^\G)<_4 \top_1,\top_2$. The order $\leq_5$ is similar to $\leq_4$ but with swapped places for $a$ and $b$. 
In $\leq_6$ we have 
$(\text{elements in }V^\G)<_6(\text{elements in }\{u_a,v_b\in V^\G_a\cup V^\G_b : u_a,v_b<^{\Pos(\G)}(uv)_1\})<_6(\text{elements in }E^\G_1)<_6(\text{elements in }\{u_a,v_b\in V^\G_a\cup V^\G_b : u_a,v_b<^{\Pos(\G)}(uv)_2\})<_6(\text{elements in }E^\G_2)<_6\infty_a,\infty_b \top_1,\top_2$.
The order $\leq_7$ is similar to $\leq_7$ but with swapped places for $1$ and $2$. 
One can readily verify that the intersection of $\leq'_1, \leq'_2,\leq'_3, \leq_4,\leq_5,\leq_6,\leq_7$ is exactly $\leq^{\Pos(\G)}$, and thus the dimension of $\Pos(\G)$ is at most 7.

Note that inserting $\bot$ to the orders $\leq'_1, \leq'_2,\leq'_3, \leq_4,\leq_5,\leq_6,\leq_7$ as the minimum element yields 7 linear extensions whose intersection is $\leq^{\Pos_\bot(\G)}$, so the dimension of $\Pos_\bot(\G)$ is also at most 7. 
\end{proof}

Combining these lemmas with the hardness result of B\'ilka et al. in \cite[Theorem 2]{BJK2011}, we obtain the following.

\begin{corollary}
The problem {\sf LogContain} {\rm(}resp., {\sf SPMorph}{\rm)} is NP-complete for inputs $\P,\Q$ of dimension at most 7, and its every element {\rm(}except for the root in case of {\sf LogContain}{\rm)} has at most 4 immediate successors and at most 12 successors in all, and $\Q = \Pos_\bot(\K_4)$ {\rm(}resp.  $\Q = \Pos(\K_4)${\rm)}. %In particular, when $\Q$ is of constant size.
\end{corollary}

Now let us turn our attention to the case that $\H$ is not considered fixed.
Chaplick et al.~\cite[Theorem 1]{CHAPLICK2015} showed that \lshomo remains NP-hard for inputs $ \G,\H$, where $\G$ is of pathwidth at most 4 and $\H$ is of pathwidth at most 3.
We observe that the construction in Section~\ref{sec:construction} preserves pathwidth, up to a constant factor.
Here, by the pathwidth of a poset we mean the pathwidth of its cover graph, i.e., we do not include edges that {\it follow} by transitivity. Formally, $\G$ is the \emph{cover graph} of a poset $\P$ if $V^\G=X^\P$ and $xy\in E^\G$ if $x\prec^\P y$ or $y\prec^\P x$. (For information on pathwidth, see \cite[Chap. 12]{Diestel}.)

\begin{lemma}
If $\G$ has pathwidth $k \geq 1$, then $\Pos(\G)$ has pathwidth at most $3k + 7$ and  $\Pos_\bot(\G)$ has pathwidth at most $3k + 8$.
\end{lemma}
\begin{proof}
Consider a path decomposition $\mathcal{P}$ of $\G$ with width at most $k$, i.e., each of its bag has at most $k+1$ elements.
We modify it into a path decomposition of $\Pos(\G)$.
First, consider an edge $e=uv$ of $\G$ and choose one bag $X^{e}$ of  $\mathcal{P}$ that contains both $u$ and $v$; it exists as  $\mathcal{P}$ is a path decomposition. We create two new bags $X^{e}_1 = X^{e}_1 \cup \{e_1\}$ and $X^{e}_2 = X^{e}_2 \cup \{e_2\}$, and insert them to  $\mathcal{P}$ immediately after $X_{e}$.
We repeat this for every edge $e \in E^{\G}$, always picking for $X_e$ a bag of the original decomposition  $\mathcal{P}$.

We denote the obtained sequence of bags by $\mathcal{P}'$.
Next, for each bag $X$ of  $\mathcal{P}'$, we replace it by $X \cup \{v_a, v_b \colon v \in X\}$.
Finally, we add the elements $\top_1,\top_2,\infty_a,\infty_b$ to every bag.
It is straightforward to verify that the obtained sequence is a path decomposition of $\Pos(\G)$ and each of its bag has at most 
\[
3(k+1) + 1 + 4 = 3k+8
\]
elements. Thus, $\Pos(\G)$ has pathwidth at most $3k+7$.

To obtain a path decomposition of $\Pos_\bot(\G)$, we need to insert $\bot$ into every bag, increasing its size by 1.

\end{proof}

This, combined with the aforementioned hardness result, yields the following.

\begin{corollary}\label{cor:hardpw}
The problem {\sf LogContain} {\rm(}resp., {\sf SPMorph}{\rm)} is NP-complete for inputs $\P,\Q$ when $\P$ has pathwidth at most 20 {\rm(}resp., 19{\rm)} and $\Q$ has pathwidth at most 17 {\rm(}resp., 16{\rm)}.
\end{corollary}

\section{Algorithm for trees}

In~\cite[Corollary 8]{FP2010} Fiala and Paulusma provided a polynomial-time algorithm that checks, for given finite graphs $\G$ and $\H$, where $\G$ is a tree, if there exists a locally surjective homomorphism from $\G$ to $\H$.  This section contains a poset analog of this result.

A finite poset $\T$ is a tree if it has a root and its every principal downset ${\downarrow}^\T t$ forms a chain.

\begin{theorem}
\label{thm:: tree}
There is a polynomial-time algorithm answering {\sf SPMorph}
and {\sf LogContain} for given finite posets $\T$ and $\Q$, where $\T$ is a tree.
\end{theorem}

We split the proof of Theorem~\ref{thm:: tree} into a series of lemmas. Here the posets $\T$ and $\Q$ are fixed. We also assume that $\Q$ has a root. We do so, as otherwise there is no surjective $p$-morphism from $\T$ onto $\Q$ anyway.

For $t\in X^\T$ let 
\[
Q_t=\{q\in X^\Q : \text{there exists a surjective  $p$-morphism from }{\uparrow}^\T t\text{ onto }{\uparrow}^\Q q\}.
\]  
Clearly, there exists a surjective $p$-morphism from $\T$ onto $\Q$ iff the root of $\Q$ is in $Q_{r}$, where $r$ is the root of $\T$.  The algorithm we are going to present computes sets $Q_t$ for $t\in X^\T$ recursively. 
Clearly, for a leaf $t$ in $\T$ (i.e., for a maximal element of $\T$) the set $Q_t$ consists of maximal elements in $\Q$. Once we have computed $Q_s$ for every immediate successor $s$ of $t$, we can compute  $Q_t$. When we find the root of $\Q$ in one of $Q_t$, we can answer the question in {\sf SPMorph} affirmatively. Indeed, by the next lemma, then the root of $\Q$ also belongs to $Q_r$, where $r$ is a root of $\T$. 

\begin{lemma}
\label{lem:: monotonicity}
Let $s,t$ be elements in $\T$ and $p,q$ be elements in $\Q$. Suppose that $s\leq^\T t$, $p\leq^\Q q$ and $p\in Q_t$. Then  $q\in Q_s$.
\end{lemma}

\begin{proof}
Let us suppose that there exists a surjective $p$-morphism $h\colon {\uparrow}^\T t\to {\uparrow}^\Q p$.  By (BP), there exists $t'\geq t$ such that $h(t')=q$. Let $u$ be any maximal element of ${\uparrow}^\Q q$. 
We define a mapping
$k\colon {\uparrow}^\T s\to {\uparrow}^\Q q$  by
\[
k(x)=
\begin{cases}
h(x) & \text{ if } x\geq^\T t',\\
q & \text{ if } x\leq^\T  t',\\
u & \text{ otherwise}.
\end{cases}
\] 
By (BP), $k$ is surjective. Let us check that $k$ is a homomorphism. It is clear that (HP) holds for every $x<y$, where $x\in  {\uparrow}^\T t' \cup {\downarrow}^\T t'$. Otherwise, $k(x)=u$. Then, since $\T$ is a tree, $y\not\in {\uparrow}^\T t' \cup {\downarrow}^\T t'$. Hence, $k(y)=u$ and (HP) holds also in this case. The satisfaction of (BP) by $k$ follows from the satisfaction of (BP) by $h$. 
\end{proof}

In what follows, it will be convenient to have a notation for the set of immediate successors of $p$ in the poset $\P$. Let us denote this set by $\isucc^\P(p)$. 

Suppose that $t$ is an element of $\T$, and we have already computed all sets $Q_{s}$, where $s\in\isucc^\T(t)$. By Lemma~\ref{lem:: monotonicity}, $\bigcup\{ Q_{s} : s\in\isucc^\T(t)\}\subseteq Q_t$. Besides this, it appears that we only need to check if $q\in Q_t$ for those $q$ that have all successors in $\bigcup\{ Q_{s} : s\in\isucc^\T(t)\}$.

\begin{lemma}
\label{lem:: restriction of Qt}
Let $t$ be an element of $\T$ and  $q$ be en element of $\Q$. If $q\in Q_t$ then 
$\isucc^\Q(q)\subseteq \bigcup\{ Q_{s} : s\in\isucc^\T(t)\}$.
\end{lemma}

\begin{proof}
Let $h\colon {\uparrow}^\T t\to {\uparrow}^\Q q$ be a surjective $p$-morphism. Let $p\in\isucc^\Q(q)$. By (BP), there exists $t'\in {\uparrow}^\T t\setminus\{t\}$ such that $h(t')=p$. Hence, $p\in Q_{t'}$. Let $s$ be the immediate successor of $t$ such that $s\leq^\T t'$. By Lemma~\ref{lem:: monotonicity}, $p\in Q_{s}$. 
\end{proof}

Thus, having computed $Q_{s}$ for every immediate successor $s$ of $t$, it only remains to decide which elements in $X^\Q\setminus \bigcup\{ Q_{s} : s\in\isucc^\T(t)\}$ and with all immediate successors in $\bigcup\{ Q_{s} : s\in\isucc^\T(t)\}$, belong to $Q_t$. We do it with the help of Hopcroft-Karp algorithm for ﬁnding a maximum matching in a bipartite graph, see e.g.~\cite[p. 763]{Cormen2009}. 

For $t\in X^\T$ and $q\in X^\Q$, let $\G_{t,q}$ be the bipartite graph with the parts $\isucc^\T(t)$ and $\isucc^\Q(q)$. A pair $sp$ is an edge in $\G_{t,q}$ if $p\in Q_{s}$.

\begin{lemma}
\label{lem:: successors}
Let $t$ be an element of $\T$ and  $q$ be an element of $\Q$. Assume that $q\not\in \bigcup\{ Q_{s} : s\in\isucc^\T(t)\}$.
% and $\{q_1,\ldots, q_n\}\subseteq \bigcup_{i=1}^m Q_{t_i}$. 
Then $q\in Q_t$ if and only if there exists a matching in $\G_{t,q}$ with $n$ edges, where $n=|\isucc^\Q(q)|$.
\end{lemma}

\begin{proof}
Suppose that there exists a matching $M$ in $\G_{t,q}$ consisting of $n$ edges $\{s_1p_1,\ldots,s_np_n\}$. Let $h_j\colon {\uparrow}^\T s_j\to{\uparrow}^\Q s_j$ be surjective $p$-morphisms witnessing this. Let $u$ be any maximal vertex in ${\uparrow}^\Q q$. Then the mapping $h\colon {\uparrow}^\T t\to{\uparrow}^\Q q$, given by
\[
h(x)=
\begin{cases}
h_j(x) & \text{ if } x\in {\uparrow}^\T s_{j}, \\
q & \text{ if } x=t, \\
u & \text{ otherwise, }
\end{cases}
\] 
is a surjective $p$-morphism.

Let us now assume that there exists a surjective $p$-morphism $h\colon {\uparrow}^\T t\to{\uparrow}^\Q q$. For an immediate successor $p$ of $q$ there exists a successor $t_p$ of $t$ such that $h(t_p)=p$. Let $s_p$ be the unique immediate successor of $t$ such that $t\prec^\T s_p\leq^\T t_p$. 
Then, by (HP), $h(s_p)\in \{q,p\}$. But, by the assumption that
$q\not\in \bigcup\{ Q_{s} : s\in\isucc^\T(t)\}$, we have 
$h(s_p)\neq q$. This shows that $p\in Q_{s_p}$. Also, since $h(s_p)=p$, we have $s_p\neq s_{p'}$ if $p\neq p'$. Consequently, $\{s_{p}p : p\in \isucc^\Q(q)\}$ is a matching in $\G_{t,q}$ with $n$ edges. 
\end{proof}

Lemmas~\ref{lem:: monotonicity},~\ref{lem:: restriction of Qt},~\ref{lem:: successors} and the presented discussion shows that the following Algorithm \ref{alg:: for tree} correctly answers the question in {\sf SPMorph}. By Corollary~\ref{cor:: from Jankov} and Lemma~\ref{lem:: monotonicity}, it also correctly answers the questions in {\sf LogContain}.  
\begin{algorithm}
\label{alg:: for tree}
\KwIn{A finite tree $\T$ and a poset $\Q$}
\lIf{$\Q$ is not rooted}{Return No}
$S\leftarrow$ the set of leaves in $T$\\
\For{$t\in S$}{
$Q_t\leftarrow$ the set of maximal elements in $\Q$
}
\For{$t$ maximal in $\T\setminus S$}{
$Q_t\leftarrow \bigcup\{ Q_{s} : s\in\isucc^\T(t)\}$\\
\For{$q$ in $\Q$ such that $\isucc^\Q(q)\subseteq\bigcup\{ Q_{s} : s\in\isucc^\T(t)\}\not\ni q$}{
\If{$G_{t,q}$ has a $|\isucc^\Q(q)|$-matching \tcc{\hfill\hfill Hopcroft-Karp algorithm}}
{
\lIf{$q$ is the root of $\Q$}{\Return Yes}
\lElse{$Q_t\leftarrow Q_t\cup\{q\}$}}
}
$S\leftarrow S\cup\{t\}$
}
\Return No
\caption{Solving {\sf SPMorph} \& {\sf LogContain} for a tree $\T$ and a poset~$\Q$}
\end{algorithm}
Notice that Algorithm~\ref{alg:: for tree} visits each element $t$ in $\T$ at most once and, for  fixed $t$, each vertex in $\Q$ at most once. The Hopcroft-Karp algorithm, for $\G_{t,q}$, runs in time $O((|\isucc^\T(t)|+|\isucc^\T(t)|)^{2.5})$.  Hence, Algorithm~\ref{alg:: for tree} runs in polynomial time. 

Finally, incorporating constructions from Lemmas~\ref{lem:: monotonicity} and \ref{lem:: successors} into Algorithm~\ref{alg:: for tree} allows us to find a $p$-morphism from $\T$ onto $\Q$, if one exists, in polynomial time.

%\begin{figure}
%\includegraphics[width=\textwidth]{fig1.eps}
%\caption{A figure caption is always placed below the illustration. Please note that short captions are centered, while long ones are justified by the macro package automatically.} \label{fig1}
%\end{figure}

 \bibliographystyle{abbrv}
 \bibliography{compary}
%Author, F.: Article title. Journal \textbf{2}(5), 99--110 (2016)

%Author, F., Author, S.: Title of a proceedings paper. In: Editor,
%F., Editor, S. (eds.) CONFERENCE 2016, LNCS, vol. 9999, pp. 1--13.
%Springer, Heidelberg (2016). \doi{10.10007/1234567890}

%Author, F., Author, S., Author, T.: Book title. 2nd edn. Publisher,

%Author, A.-B.: Contribution title. In: 9th International Proceedings
%on Proceedings, pp. 1--2. Publisher, Location (2010)

%LNCS Homepage, \url{http://www.springer.com/lncs}, last accessed 2023/10/25
%\end{thebibliography}
\end{document}